\documentclass[10pt]{article}

\newcommand{\ket}[1]{|#1\rangle}

\usepackage[margin=3cm, a4paper]{geometry}
\usepackage{amsmath,amsthm,amsfonts,amssymb}
\usepackage{fontenc}
\usepackage[all]{xy}
\usepackage{graphicx,subfigure,multirow}
\usepackage[figurename=Fig.]{caption}
\usepackage{tikz}
\usetikzlibrary{shapes, arrows, shadows}
\usepackage[latin1]{inputenc}
\usepackage[scaled]{helvet}
\usepackage[toc, page]{appendix}

\newtheorem{thm}{Theorem}[subsection]
\newtheorem{define}[thm]{Definition}

\newtheorem{lem}[thm]{Lemma}
\newtheorem{prop}[thm]{Proposition}
\newtheorem{corr}[thm]{Corollary}
\newtheorem{assumption}{Assumption}
\numberwithin{equation}{subsection}

\begin{document} 

\begingroup  
\centering
{\Large\textbf{Computation in generalised probabilistic theories} \\[1.5em]
  \normalsize Ciar\'{a}n M. Lee\footnote{
Electronic address: ciaran.lee@cs.ox.ac.uk}  and Jonathan Barrett}\\[1em]
{\it  University of Oxford, Department of Computer Science, Wolfson Building, Parks Road, Oxford OX1 3QD, UK.}\\[1em]
\endgroup

\begin{abstract} 
From the general difficulty of simulating quantum systems using classical systems, and in particular the existence of an efficient quantum algorithm for factoring, it is likely that quantum computation is intrinsically more powerful than classical computation. At present, the best upper bound known for the power of quantum computation is that $\bold{BQP}\subseteq\bold{AWPP}$, where  $\bold{AWPP}$ is a classical complexity class (known to be included in $\bold{PP}$, hence $\bold{PSPACE}$). This work investigates limits on computational power that are imposed by simple physical, or information theoretic, principles. To this end, we define a circuit-based model of computation in a class of operationally-defined theories more general than quantum theory, and ask: what is the minimal set of physical assumptions under which the above inclusions still hold? We show that given only an assumption of tomographic locality (roughly, that multipartite states and transformations can be characterised by local measurements), efficient computations are contained in $\bold{AWPP}$. This inclusion still holds even without assuming a basic notion of causality (where the notion is, roughly, that probabilities for outcomes cannot depend on future measurement choices). Following Aaronson, we extend the computational model by allowing post-selection on measurement outcomes. Aaronson showed that the corresponding quantum complexity class, $\bold{PostBQP}$, is equal to $\bold{PP}$. Given only the assumption of tomographic locality, the inclusion in $\bold{PP}$ still holds for post-selected computation in general theories. Hence in a world with post-selection, quantum theory is optimal for computation in the space of all operational theories. We then consider whether one can obtain relativised complexity results for general theories. It is not obvious how to define a sensible notion of a computational oracle in the general framework that reduces to the standard notion in the quantum case. Nevertheless, it is possible to define computation relative to a `classical oracle'. Then, we show there exists a classical oracle relative to which efficient computation in any theory satisfying the causality assumption does not include $\bold{NP}$. 
\end{abstract}
 
\section{Introduction} 

Quantum theory offers dramatic new advantages for various information theoretic tasks \cite{Nielsen}. This raises the general question of what broad relationships exist between physical principles, which a theory like quantum theory may or may not satisfy, and information theoretic advantages.
Much progress has already been made in understanding the connections between physical principles and some tasks, such as cryptography and communication complexity problems. It is now known that the degree of non-locality in a theory is related to its ability to solve communication complexity problems \cite{PR} and to its ability to perform super-dense coding, teleportation and entanglement swapping \cite{DensePR}. Teleportation and no-broadcasting are now better understood than they were when investigated solely from the viewpoint of quantum theory \cite{tele, broadcast}. Cryptographic protocols have been developed whose security relies not on aspects of the quantum formalism, but on general physical principles. For example, device-independent key distribution schemes have been developed that are secure against attacks by post-quantum eavesdroppers limited only by the no-signalling principle \cite{Crypt}. 

By comparison, relatively little has been learned about the connections between physical principles and computation. It was shown in \cite{Rev} that a maximally non-local theory has no non-trivial reversible dynamics and, thus, any reversible computation in such a theory can be efficiently simulated on a classical computer. Aside from this result, most previous investigations into computation beyond the usual quantum formalism have centred around non-standard theories involving modifications of quantum theory. These theories often appear to have immense computational power and entail unreasonable physical consequences. For example, non-linear quantum theory appears to be able to solve $\bold{NP}$-complete problems in polynomial time \cite{Non}, as does quantum theory in the presence of closed timelike curves \cite{CTC, JMA}. Aaronson has considered other modifications of quantum theory, 
such as a hidden variable model in which the history of hidden states can be read out by the observer \cite{sample}, and these have also been shown to entail computational speedups over the usual quantum formalism. 
 
This work considers computation in a framework suitable for describing essentially arbitrary operational theories, where an operational theory specifies a set of laboratory devices that can be connected together in different ways, and assigns probabilities to experimental outcomes. Theories within this framework can be described that are different from classical or quantum theories, but which nonetheless make good operational sense and do not involve peculiarities like closed timelike curves. The framework, described in Section~\ref{framework} suggests a natural model of computation, analogous to the classical and quantum circuit models, described in Section~\ref{compmodel}. 

The strongest known non-relativised upper bound for the power of quantum computation is that the class $\bold{BQP}$ of problems efficiently solvable by a quantum computer is contained in the classical complexity class $\bold{AWPP}$. The class $\bold{AWPP}$ has a slightly obscure definition, but is well known to be contained in $\bold{PP}$, hence $\bold{PSPACE}$. Section~\ref{upperboundscomppower} shows that the same result holds for any theory in the operational framework that satisfies the principle of tomographic locality, where this means, roughly, that transformations can be completely characterised by product states and effects. That is, if the complexity class of problems that can be efficiently solved by a specific theory $\bold{G}$ is denoted schematically $\bold{BGP}$, then for tomographically local theories, $\bold{BGP} \subseteq \bold{AWPP}$. Once suitable definitions are in place, the proof is essentially the same as the proof for the quantum case: the idea is that this proof can be cast in a theory-independent manner, and be seen to follow from a very minimal set of assumptions on the structure of a physical theory. In fact, the containment $\bold{BGP} \subseteq \bold{AWPP}$ still holds even in the absence of a basic principle of causality (which, if it does hold, ensures that there can be no signalling from future to past). 


It was suggested in \cite{Jon} that quantum theory achieves, in some sense, an optimal balance between its set of states and its dynamics, and that this balance entails that quantum theory is powerful for computation by comparison with most theories in the space of operational theories. Although the status of this suggestion is unknown, it turns out to be exactly correct in the context of a world allowing post-selection of measurement outcomes. Aaronson showed that the class of problems efficiently solvable by a quantum computer with the ability to post-select measurement outcomes is equal to the class $\bold{PP}$ \cite{Post}. Section~\ref{post} extends the idea of computation with post-selection to general theories, and shows that given (as always) tomographic locality, problems efficiently solvable by any theory with post-selection are contained in $\bold{PP}$. In other words: any problem efficiently solvable in a tomographically local theory with post-selection, is also efficiently solvable by a quantum computer with post-selection.

Finally, oracles play a special role in quantum computation, forming the basis of most known computational speed-ups over classical computation. Section~\ref{oracle} discusses the problem of defining a sensible notion of oracle in the general framework, that reduces to the standard definition in quantum theory. This problem may not have a solution that is completely general, hence we introduce instead a notion of `classical oracle' that can be defined in any theory that satisfies the causality principle. There then exists a classical oracle such that relative to this oracle, $\bold{NP}$ is not contained in $\bold{BGP}$ for any theory $\bold{G}$ satisfying tomographic locality and causality . 

\section{The framework}\label{framework}   

We will work in the circuit framework for generalised probabilistic theories developed by Hardy in \cite{Hardy1, Hardy2} and Chiribella, D'Ariano and Perinotti in \cite{Pavia1,Pavia2}. The presentation here is most similar to that of Chiribella et al. 

\subsection{Tests and circuits} \label{intro}

The idea of a generalised probabilistic theory is that a set of physical, or laboratory, devices is specified, which can be connected together in different ways, such that the theory will give probabilities for different outcomes. Such theories take \emph{tests} as their primitive notions, where a test can be thought of as corresponding to a physical device with input ports, output ports, and a classical pointer. Whenever the test is applied, the pointer ends up in one of a number of positions indicating a classical outcome. Input and output ports are typed, with types given by labels $A,B,C\dots$. As discussed in more detail below, tests can be composed both sequentially and in parallel, and when tests are composed sequentially, types must match: the output ports of the first device must have the same types as the corresponding input ports of the second.

Suppose that for a particular test, the classical outcome $r$ takes values in a set $X$. We shall assume throughout that $|X|$ is finite. A test ${\mathcal E}$, with specified input and output types, then defines a set of \emph{events}, one for each classical outcome, $\{\mathcal{E}_r\}_{r\in{X}}$. With an input port of type $A$ and an output port of type $B$, for example, the test can be represented diagrammatically as
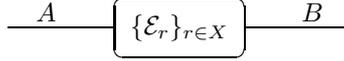
\[
\begin{xy}
(22.5,0)*++[0][F-:<3pt>]{\text{$\{\mathcal{E}_r\}_{r\in{X}}$}}="1";
(0,0);"1"**\dir{-};
(5,2)*{\text{$A$}};
"1";(45,0)**\dir{-};
(40,2)*{\text{$B$}};
\end{xy}$$
and a specific event as
$$
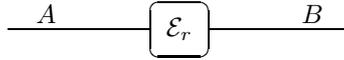
\begin{xy}
(22.5,0)*++[0][F-:<3pt>]{\text{$\mathcal{E}_{r}$}}="1";
(0,0);"1"**\dir{-};
(5,2)*{\text{$A$}};
"1";(45,0)**\dir{-};
(40,2)*{\text{$B$}};
\end{xy}
\]
A test is \emph{deterministic} if its outcome set $X$ is the singleton set. 

Although tests, with input and output ports, and a pointer, form the primitives of the operational theory, it is also useful to introduce a notion of \emph{physical system}. A system may be thought of as passing between the output port of a device, and the input port of the next, and has the same type as the ports. In other words, in the diagrams above and below, systems correspond to wires. Given two systems of types $A$ and $B$, we can form a \emph{composite system} of type $AB$. Operationally, a test with input system $AB$ corresponds to a physical device with a set of input ports labelled by $A$ and a disjoint set of input ports labelled by $B$. 

A test with no input ports corresponds to a preparation of a system -- more precisely, such a test corresponds to a set of preparations, with the classical pointer indexing which preparation actually occurs. Such a test can be represented diagrammatically as:
$$
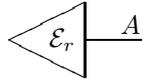
\begin{xy}
(9,0)*{\text{$\mathcal{E}_r$}};
(12,5);(12,-5)**\dir{-};
(12,5);(2,0)**\dir{-};
(12,-5);(2,0)**\dir{-};
(12,0);(20,0)**\dir{-};
(18,2)*{\text{$A$}};
\end{xy}$$
 A test with no output ports corresponds to a measurement (that destroys or discards the system), with the classical pointer indexing the measurement outcome. Diagrammatically, such a test can be written:
$$
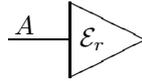
\begin{xy}
(12,0)*{\text{$\mathcal{E}_r$}};
(9,5);(9,-5)**\dir{-};
(9,5);(19,0)**\dir{-};
(9,-5);(19,0)**\dir{-};   
(1,0);(9,0)**\dir{-};
(3,2)*{\text{$A$}};
\end{xy}$$

Both tests and events can be composed in sequence and in parallel. If $\{\mathcal{E}_{r_1}\}_{r_1\in{X_1}}$ is a test from system $A$ to $B$ and $\{\mathcal{U}_{r_2}\}_{r_2\in{X_2}}$ is a test from system $B$ to $C$, then their sequential composition is a test from $A$ to $C$ with outcomes $(r_1,r_2)\in{X_1\times{X_2}}$ and events $\{\mathcal{U}_{r_2}\circ\mathcal{E}_{r_1}\}_{(r_1,r_2)\in{X_1\times{X_2}}}$. Similarly,  if $\{\mathcal{E}_{r_1}\}_{r_1\in{X_1}}$ is a test from system $A$ to $B$ and $\{\mathcal{U}_{r_2}\}_{r_2\in{X_2}}$ is a test from system $C$ to $D$, then their parallel composition is a test from the composite system $AC$ to the composite system $BD$ with outcomes $(r_1,r_2)\in{X_1\times{X_2}}$ and events $\{\mathcal{U}_{r_2}\otimes\mathcal{E}_{r_1}\}_{(r_1,r_2)\in{X_1\times{X_2}}}$. 
Sequential and parallel composition satisfy 
$$\big(\mathcal{U}_{r_3}\otimes\mathcal{E}_{r_4}\big)\circ\big(\mathcal{F}_{r_1}\otimes\mathcal{K}_{r_2}\big)=\big(\mathcal{U}_{r_3}\circ\mathcal{F}_{r_1}\big)\otimes\big(\mathcal{E}_{r_4}\circ\mathcal{K}_{r_2}\big)$$
for every $\mathcal{U}_{r_3}, \mathcal{E}_{r_4}, \mathcal{F}_{r_1}, \mathcal{K}_{r_2}$ with the property that the output of $\mathcal{F}_{r_1}$ (respectively, $\mathcal{K}_{r_2}$) matches the input of $\mathcal{U}_{r_3}$ (respectively, $\mathcal{E}_{r_4}$).  A generalised probabilistic theory specifies a set of tests, closed under sequential and parallel composition. 

A circuit in a generalised probabilistic theory corresponds to a number of tests, connected in sequence and in parallel, such that there are no unconnected ports (i.e., no dangling input or output wires), and no cycles. \footnote{Connected sets of tests with dangling wires may be called \emph{open circuits}, but this work has no need to consider open circuits, so we use the term \emph{circuit} throughout to refer to a closed circuit.}  
For example: 
$$
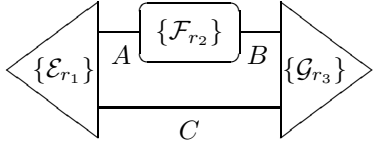
\begin{xy}
(24,5)*++[0][F-:<3pt>]{\text{$\{\mathcal{F}_{r_2}\}$}}="1";
"1";(12,5)**\dir{-};
(7.5,0)*{\text{$\{\mathcal{E}_{r_1}\}$}};
(12,-5);(36,-5)**\dir{-};
(24,-8)*{\text{$C$}};
(12,9);(12,-9)**\dir{-};
(0,0);(12,9)**\dir{-};
(0,0);(12,-9)**\dir{-};
(15,2)*{\text{$A$}};
"1";(36,5)**\dir{-};
(33,2)*{\text{$B$}};
(36,9);(36,-9)**\dir{-};
(36,9);(48,0)**\dir{-};
(36,-9);(48,0)**\dir{-};
(40.5,0)*{\text{$\{\mathcal{G}_{r_3}\}$}};
\end{xy}$$
A specific outcome of the above circuit corresponds to a particular classical outcome for each of the tests, i.e., to a collection of events, connected in sequence and in parallel:
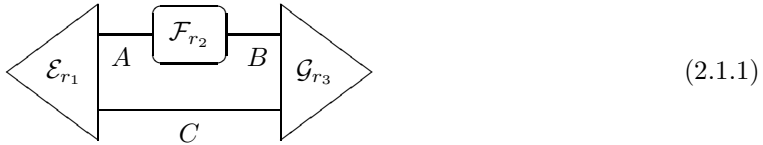
\begin{equation}\label{circuit outcome}
\begin{xy}
(24,5)*++[0][F-:<3pt>]{\text{$\mathcal{F}_{r_2}$}}="1";
"1";(12,5)**\dir{-};
(7.5,0)*{\text{$\mathcal{E}_{r_1}$}};
(12,-5);(36,-5)**\dir{-};
(24,-8)*{\text{$C$}};
(12,9);(12,-9)**\dir{-};
(0,0);(12,9)**\dir{-};
(0,0);(12,-9)**\dir{-};
(15,2)*{\text{$A$}};
"1";(36,5)**\dir{-};
(33,2)*{\text{$B$}};
(36,9);(36,-9)**\dir{-};
(36,9);(48,0)**\dir{-};
(36,-9);(48,0)**\dir{-};
(40.5,0)*{\text{$\mathcal{G}_{r_3}$}};
\end{xy}
\end{equation}

\subsection{Probabilistic structure}

So far, we have described the operational part of a generalised probabilistic theory, but not the probabilistic part. In addition to specifying a set of tests, hence sets of circuits and circuit outcomes, a probabilistic theory should assign probabilities to circuit outcomes. In a generalised probabilistic theory, every outcome of a circuit is assigned a probability $P(r_1r_2\ldots r_n)$, understood as the joint probability of outcomes $r_1,\ldots ,r_n$ for the individual tests occurring on a single run. The joint probabilities satisfy $\sum_{r_1 r_2\ldots r_n} P(r_1r_2\ldots r_n) = 1$. 
A further constraint is that probabilities for unconnected, i.e., independent, circuits factorise. This means that for events $\mathcal{E}_{r_1r_2\ldots r_m}$ and $\mathcal{F}_{s_1s_2\ldots s_n}$, each of which corresponds to the outcome of a closed circuit, probabilities assigned to the composite events $\mathcal{E}_{r_1r_2\ldots r_m} \otimes \mathcal{F}_{s_1s_2\ldots s_n}$, and $\mathcal{E}_{r_1r_2\ldots r_m} \circ \mathcal{F}_{s_1s_2\ldots s_n}$, each satisfy $P(r_1\ldots r_m,s_1\ldots s_n) = P(r_1\ldots r_m) P(s_1\ldots s_n)$.

The introduction of probabilities into the theory induces linear structure that will be crucial in what follows. Consider two events $\mathcal{E}_0$ and $\mathcal{E}_1$, whose input and output ports have matching types. Suppose that for every closed circuit, and every outcome of the circuit, replacing $\mathcal{E}_0$ with $\mathcal{E}_1$ does not change the probability of the outcome. In this case, $\mathcal{E}_0$ and $\mathcal{E}_1$ are \emph{equivalent}. The events $\mathcal{E}_0$ and $\mathcal{E}_1$ may be easily distinguished operationally by the fact that the corresponding physical devices look quite different, but there is no distinction between $\mathcal{E}_0$ and $\mathcal{E}_1$ from the point of view of the probabilistic predictions of the theory. We refer to the equivalence classes of events formed in this way as \emph{transformations}. The following will mostly be concerned with transformations, rather than the underlying primitive events. Transformations with no input ports we will sometimes call \emph{states}, and transformations with no output ports, \emph{effects}. For system types $A$ and $B$, the sets of transformations from $A$ to $B$, states on $A$ and effects on $B$ are denoted $\bold{Transf(A,B)}$, $\bold{St(A)}$, and $\bold{Eff(B)}$ respectively.

Quantum theory provides a specific example of a theory that can be described in this framework. A system is associated with a complex Hilbert space, with the type of the system given by the dimension of the Hilbert space. States and effects are associated with positive operators, and transformations are associated with trace non-increasing completely positive maps. A test with no input ports corresponds to what is sometimes called a `random source of quantum states', and is associated with positive operators $\{\rho_r\}$ such that $\sum_r\mathrm{Tr}(\rho_r)=1$. When the test is performed, the probability that the classical pointer takes position $r$ is given by $\mathrm{Tr}(\rho_r)$, and the quantum state that is prepared, conditioned on the pointer reading being $r$, is the normalised operator $\rho_r / \mathrm{Tr}(\rho_r)$. A test with no output ports is associated with a positive operator-valued measurement, that is a set of positive operators $\{E_i\}$ satisfying $\sum_i E_i=\mathbb{I}$. A test with both input and output ports is associated with a \emph{quantum instrument}, that is a set of trace non-increasing completely positive maps, one for each value of the pointer reading $r$, that sum to a trace-preserving map. Given these associations, the standard rules of quantum theory allow the probability to be calculated for any circuit outcome.

Returning to the general framework, it is convenient to use the `Dirac-like' notation $|\sigma_{r})_A$ to represent a state of system type $A$, and $_A(\lambda_r|$ to represent an effect on system type $A$, so that if the state $|\sigma_{r_1})_A$ is followed by the effect $_A(\lambda_{r_2}|$, the joint probability of obtaining outcome $r_1$ for the preparation and outcome $r_2$ for the measurement is given by
\[
_A(\lambda_{r_2}|\sigma_{r_1})_A := P(r_1,r_2).
\]
In the following, we shall sometimes drop the input/output type label. A state $|\sigma_{r_1})_A$ can be identified with a function from effects on $A$ into probabilities, such that 
\[
_A(\lambda_{r_2}| \mapsto {_A(\lambda_{r_2}}|\sigma_{r_1})_A.
\]
Since one can take linear combinations of functions, the set of states $\bold{St(A)}$ can be extended to a real vector space, which we denote $\bold{V_A}$. In quantum theory, for example, states are positive operators, which span the real vector space $\bold{V_A}$ of Hermitian operators.

Similarly, an effect $_A(\lambda_{r_2}|$ can be identified with a function from preparation events to probabilities: 
\[
|\sigma_{r_1})_A \mapsto {_A(\lambda_{r_2}}|\sigma_{r_1})_A,
\]
and the set of effects $\bold{Eff(A)}$ can be extended to a real vector space $\bold{V^A}$. A more general kind of transformation, from (possibly composite) system type $A$ to (possibly composite) system type $B$, defines a function into probabilities, where the domain is now circuit fragments with the property that there are unconnected input and output ports, such that adding in a transformation of this type results in a closed circuit. Again, this means that the set of transformations $\bold{Transf(A,B)}$ can be extended to a real vector space, denoted $\bold{V_B^A}$.

Throughout the paper, we adopt
\begin{assumption}\label{finitedim}
For every pair of system types $A$ and $B$, and every transformation from $A$ to $B$, $\bold{V^A_B}$ is finite dimensional.
\end{assumption}
As a consequence, the vector space generated by effects on a system can be regarded as dual to the space of states, and vice versa: $V^A = (V_A)^*$ and $V_A = (V^A)^*$. In other works on generalised probabilistic theories, it is quite often assumed that the sets $\bold{Transf(A,B)}$, $\bold{St(A)}$, and $\bold{Eff(B)}$ are convex subsets of the corresponding vector spaces, the idea being that probabilistic mixtures of allowed transformations should also be allowed transformations. This work, however, doesn't need this assumption: the main constraints on sets of transformations, states and effects are closure under sequential and parallel composition.

\subsection{Tomographic locality} \label{LT} 

Every transformation $T_s$ from $A$ to $B$ induces a linear map from $\bold{V_A}$ to $\bold{V_B}$, uniquely defined by
\begin{equation}\label{linearmap}
|\sigma_r)_{A}\in\bold{St(A)} \mapsto T_s |\sigma_r)_{A} \in \bold{St(B)},
\end{equation}
where $T_s |\sigma_r)_{A}$ is the state of type $B$, corresponding to composition of $T_s$ with $|\sigma_r)_A$. Without further assumptions, however, this map is in general \emph{not} sufficient to specify the transformation $T_s$. To see this,  consider the situation in which the transformation $T_s$ is applied to one half of a bipartite state $|\sigma)_{AC}$. The composition defines a bipartite state of type $BC$, which can be schematically represented $|\sigma')_{BC} = (T_s \otimes I_C) |\sigma)_{AC}$, with $I_C$ understood as an identity transformation (or the absence of any transformation) on system $C$. The action of $T_s$ on bipartite states of type $AC$ induces a linear map from $\bold{V_{AC}}$ to $\bold{V_{BC}}$. In general, however, there need be no simple relationship between this map, and the map above from $\bold{V_A}$ to $\bold{V_B}$. Indeed, there need not be any simple relationship between the vector space $\bold{V_{AC}}$ and the vector spaces for the individual systems, $\bold{V_A}$ and $\bold{V_C}$. For each possible system type $C$, this structure is ultimately specified by the theory, via the assignments of probabilities to circuit outcomes.\footnote{The operational content of Assumption~\ref{finitedim} is that there does at least exist a finite set of system types $C$, such that specification of the action of $T_s\otimes I_C$ on $\bold{V_{AC}}$ for each of the system types in this finite set is sufficient to characterise $T_s$.}

The representation of transformations in a generalised probabilistic theory is greatly simplified by the assumption of \emph{tomographic locality}. A theory satisfies tomographic locality if every transformation can be fully characterized by local process tomography. More formally, consider transformations $T^1_{s_1}$ and $T^2_{s_2}$, both of which have input type $A_1\cdots A_m$ and output type $B_1 \cdots B_n$. Consider circuit outcomes of the form
$$
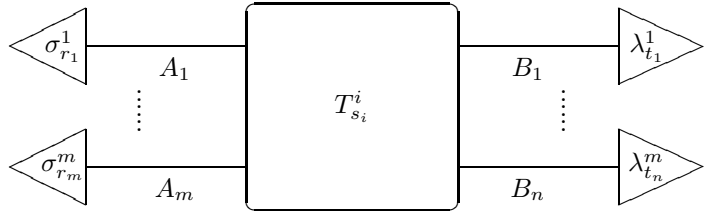
\begin{xy}
(47,1)*+++++++++++[0][F-:<3pt>]{\text{$T^i_{s_i}$}}="2";
(9,9)*{\text{$\sigma^1_{r_1}$}};
(12,-7);(33,-7)**\dir{-};
(12,9);(33,9)**\dir{-};
(23.7,-10)*{\text{$A_m$}};
(12,14);(12,4)**\dir{-};
(2,9);(12,14)**\dir{-};
(2,9);(12,4)**\dir{-};
(9,-7)*{\text{$\sigma^m_{r_m}$}};
(12,-2);(12,-12)**\dir{-};
(2,-7);(12,-2)**\dir{-};
(2,-7);(12,-12)**\dir{-};
(23.7,6)*{\text{$A_1$}};
(19,3)*{\text{.}};
(19,2)*{\text{.}};
(19,1)*{\text{.}};
(19,0)*{\text{.}};
(19,-1)*{\text{.}};
(19,-2)*{\text{.}};
(61,9);(82.1,9)**\dir{-};
(70.1,6)*{\text{$B_1$}};
(70.1,-10)*{\text{$B_n$}}; 
(74.8,3)*{\text{.}};
(74.8,2)*{\text{.}};
(74.8,1)*{\text{.}};
(74.8,0)*{\text{.}};
(74.8,-1)*{\text{.}};
(74.8,-2)*{\text{.}};
(85.8,9)*{\text{$\lambda^1_{t_1}$}};
(82.1,14);(82.1,4)**\dir{-};
(93.1,9);(82.1,14)**\dir{-};
(93.1,9);(82.1,4)**\dir{-};
(85.8,-7)*{\text{$\lambda^m_{t_n}$}};
(82.1,-12);(82.1,-2)**\dir{-};
(93.1,-7);(82.1,-12)**\dir{-};
(93.1,-7);(82.1,-2)**\dir{-};
(61,-7);(82.1,-7)**\dir{-};
\end{xy}$$ 
with corresponding probability $P^i(r_1\ldots r_m,t_1\ldots t_n, s_i)$, where $i\in \{1,2\}$. Tomographic locality states that for all transformations $T^1_{s_1}$ and $T^2_{s_2}$ with matching input and output types, if 
\[
P^1(r_1\ldots r_m,t_1\ldots t_n, s_1) = P^2(r_1\ldots r_m,t_1\ldots t_n, s_2) \quad \forall |\sigma^1_{r_1}),\ldots , |\sigma^n_{r_m}), (\lambda^1_{t_1}|,\ldots , (\lambda^n_{t_n}| 
\]
then
\[
T^1_{s_1} = T^2_{s_2}.
\]

The whole of the rest of this work adopts
\begin{assumption}\label{tomloc}
Tomographic locality is satisfied.
\end{assumption}
A consequence of tomographic locality is that for a transformation with input type $AB$ and output type $CD$, the corresponding real vector space has the form \cite{Jon, Pavia1,Pavia2} ,
\[
\bold{V^{AB}_{CD}} \cong \bold{V^A} \otimes \bold{V^B} \otimes \bold{V_C} \otimes \bold{V_D},
\]
where $\otimes$ here denotes the ordinary vector space tensor product (as opposed to the symbolic $\otimes$ used above to denote parallel composition). In particular, for a bipartite state of type $AC$, the corresponding vector space $\bold{V_{AC}} \cong \bold{V_A} \otimes \bold{V_C}$. Furthermore, a transformation $T_s\in\bold{Transf(A,B)}$ is completely specified by its action on $\bold{St(A)}$, hence $T_s$ can be identified with the linear map defined by Eq.~(\ref{linearmap}). When $T_s$ acts on part of a bipartite state of type $AC$, the induced linear map $\bold{V_{AC}}\rightarrow \bold{V_{BC}}$ is given by $T_s\otimes I_C$, where again, the symbol $\otimes$ represents the ordinary vector space tensor product, and $I_C$ is now the identity operator on the vector space $\bold{V_C}$. In view of Assumptions~\ref{finitedim} and \ref{tomloc}, the symbol $\otimes$ will from here on denote the ordinary tensor product of finite dimensional vector spaces.  

Fixing a basis for each system type, a transformation $T$ with input $AB$ and output $CD$ can be written as a matrix
$$T = \sum_{i,j,k,l} M_{ij,kl}\big(\alpha^{A}_i\otimes\alpha^{B}_j\otimes\alpha^{C}_k\otimes\alpha^{D}_l\big),$$ 
where $M_{ij,kl}\in\mathbb{R}$, $\{\alpha^{A}_i\}$, $\{\alpha^{B}_j\}$ are bases for $\bold{V^A}$ and $\bold{V^B}$ respectively, and $\{\alpha^C_l\}$, $\{\alpha^D_m\}$ are bases for $\bold{V_C}$ and $\bold{V_D}$ respectively. The probability associated with a circuit outcome, e.g., of the form of Fig.~(\ref{circuit outcome}), can be written
\[
M^3_{r_3} . (M^2_{r_2} \otimes I_C) . M^1_{r_1},
\]
where $M^1_{r_1}$ (a column vector) is the matrix form of the transformation corresponding to the event $\mathcal{E}_{r_1}$, $M^2_{r_2}$ corresponds to $\mathcal{F}_{r_2}$, and $M^3_{r_3}$ (a row vector) corresponds to $\mathcal{G}_{r_3}$. 
 
\subsection{Causality}\label{causality}  

A nice feature of the Pavia-Hardy framework we have described is that a basic assumption of causality is not implicit, but can be articulated explicitly and theories considered that do not satisfy this assumption. A generalised probabilistic theory is said to be \emph{causal} if the marginal probability of a preparation event is independent of the choice of which measurement follows the preparation. More formally, if $\{|\sigma_i)\}_{i\in{X}}\subset\bold{St(A)}$ are the states corresponding to a preparation test, consider the probability of outcome $i$, given that a subsequent measurement $\mathcal{E}$ corresponds to a set of effects $\{ {(\lambda_j|\}_{j\in{Y}}}$:
\[
P(i | \mathcal{E} ):=\sum_{j\in{Y}}(\lambda_j|\sigma_i).
\]
The theory is causal if for any system type $A$, any preparation test with outcome $i$, and any pair of measurements, $\mathcal{E}$ and $\mathcal{F}$, with input type $A$, 
\[
P(i|\mathcal{E})=P(i|\mathcal{F}).
\]
Note that the causality assumption is logically independent from tomographic locality: generalised probabilistic theories satisfying one or both or neither can be defined.

If circuits are thought of as having a temporal order, with tests later in the sequence occurring at a later time than tests earlier in the sequence, then the assumption of causality captures the intuitive notion of \emph{no signalling from the future}. It was shown in \cite{Pavia1} that a generalised probabilistic theory is causal if and only if for every system type $A$, there is a unique deterministic effect $_A(u|$. In this case, a measurement, with corresponding effects $\{ {(\lambda_j|\}_{j\in{Y}}}$, satisfies $\sum_j { (\lambda_j |} = {(u|}$. A state $|\sigma)$ is normalised if and only if ${(u|\sigma)} = 1$. The causality assumption also implies \cite{Pavia1} a \emph{no-signalling} principle for the states of the theory. That is, in a causal theory, if a test is performed on the $A$ part of a composite system of type $AB$, then it is not possible to get information about which test was performed by only performing a test on the $B$ part. (For an interesting extension of this idea to arbitrary causal networks, corresponding to circuits in the Pavia-Hardy framework, see \cite{hlp}.)

Although the idea of \emph{no-signalling from the future} seems intuitive, there is nothing obviously pathological about generalised probabilistic theories that do not satisfy the causality assumption, as long as one does not try to define adaptive circuits, wherein a choice of later test can depend on an earlier outcome. Indeed there is nothing about  the framework as it stands that forces an interpretation of the circuits described as a sequence of tests applied in a temporal order matching the order of tests in the circuit. Perhaps an entire closed circuit is set up in advance, and the pointers attain their final resting positions together, when a ``go'' button is pressed. Remarkably, the majority of the results derived in this work do not require the causality assumption, hence: \emph{except where explicitly stated, causality is not assumed in what follows}. 

\subsection{Examples}


As already noted, quantum theory can be formulated as a generalised probabilistic theory in the above framework, with finite dimensional quantum theory satisfying Assumption~\ref{finitedim}. Quantum theory satisfies the causality assumption, as the probability of an event cannot depend on the choice of a measurement that is subsequently performed on the system. For a system associated with Hilbert space $H$, the unique deterministic effect, guaranteed to exist in a theory satisfying the causality assumption, is simply the identity operator $\mathbb{I}$ on $H$. For a system of type $A$, the vector space $\bold{V_A}$ is the real vector space of Hermitian operators, spanned by the density matrices. It is well known that quantum theory satisfies the assumption of tomographic locality. This follows from the way in which systems combine to form composite systems: a joint state is a positive operator acting on the tensor product of the Hilbert spaces associated with the individual systems. One can then check that the real vector spaces of Hermitian operators satisfy $\bold{V_{AB}} \cong \bold{V_A}\otimes \bold{V_B}$. 

The framework presented is also general enough to accommodate the basic classical theory of finite dimensional probability distributions and stochastic processes, as well as probabilistic theories different from either quantum or classical theory. The latter include ``box world'' \cite{Jon, DensePR}, a causal theory allowing for arbitrarily strong nonlocal correlations, such as the PR box correlations of Popescu and Rohrlich \cite{box} that maximally violate the CHSH inequality. Quantum theory defined over real, rather than complex, Hilbert spaces supplies an example of a theory that does not satisfy tomographic locality. See also \cite{cause} for an explicit construction that does not satisfy the causality assumption.

\section{Computation in generalised probabilistic theories}\label{compmodel}

\subsection{Uniform circuits}   

The last section showed that in a generalised probabilistic theory, one can draw circuits representing the connections of physical devices in an experiment, and the specific events that took place in the experiment. These circuits provide a natural model of computation, based on the classical and quantum circuit models. A good notion of \emph{efficient} computation needs a definition of a \emph{uniform family of circuits} in a generalised probabilistic theory. 

In the standard, classical or quantum, circuit model, a circuit family $\{C_n\}=\{C_1, C_2, \dots\}$ consists of a sequence of circuits, each indexed by a positive integer $n$, denoting the input system size, where $C_n$ is the circuit corresponding to a problem instance of size $n$. In a poly-size circuit family, the number of gates in $C_n$ is bounded by a polynomial in $n$, and the circuit family is uniform if a Turing machine can output a description of $C_n$ in time bounded by a polynomial in $n$.

In a generalised probabilistic theory, there is no reason to assume that a circuit must have the form of a number of gates acting on some input, where the input preparation encodes the problem instance -- recall that we do not necessarily assume that the generalised probabilistic theory satisfies the causality assumption, in which case a circuit does not have a preferred direction. Instead, we allow the entire circuit to encode the problem instance, defining a circuit family as a set $\{C_x\}$ such that each circuit is indexed by a classical string $x=x_1x_2\dots x_n$. A circuit family is poly-size if the number of gates is bounded by a polynomial in $|x|$. For a particular generalised probabilistic theory it might not be the case that bipartite and single system transformations together are universal for computation, as they are in classical and quantum computation. Hence for any $k,l$, a circuit might involve gates with $k$ input systems and $l$ output systems. In general, it might be the case that no finite gate set is universal for computation. Nonetheless, we will impose as a requirement of uniformity that any uniform circuit family is associated with a finite gate set,\footnote{For a uniformity condition where the size of the gate set grows with circuit size, see \cite{NdB}.} such that each circuit in the family is built from elements of that set. It follows that the number of distinct system types appearing in a uniform circuit family is also finite.  

A further requirement for a circuit family to be uniform takes the form of a constraint on the entries of the matrices representing the transformations that appear in the finite gate set -- otherwise, it may be possible to smuggle hard to compute quantities into the computation. There must exist some fixed choice of basis of $\bold{V_A}$ for each system $A$, such that a Turing machine can efficiently compute approximations to the entries of the matrices relative to these bases. We require that for any matrix entry $({M})_{ij}$, and any $\epsilon$, a Turing machine can output a rational number, within $\epsilon$ of $({M})_{ij}$, in time bounded by a polynomial in $\log(\frac{1}{\epsilon})$. This is physically reasonable, since gates are supposed to represent operational devices, and it makes sense to assume that an experimenter with access to devices governed by some generalised probabilistic theory cannot align, or employ, them with arbitrary accuracy.

Finally, for a circuit family $\{C_x\}$ to be uniform, there must be a Turing machine that, acting on input $x$, outputs a classical description of $C_x$ in time bounded by a polynomial in $|x|$. 

The notion of a poly-size uniform circuit family $\{C_x\}$ can be summarised as follows:
\begin{itemize} 
\item The number of gates in the circuit $C_x$ is bounded by a polynomial in $|x|$.
\item There is a finite gate set $\mathcal{G}$, such that each circuit in the family is built from elements of $\mathcal{G}$.
\item For each type of system, there is a fixed choice of basis, relative to which transformations are associated with matrices. Given the matrix ${M}$ representing (a particular outcome of) a gate in $\mathcal{G}$, a Turing machine can output a matrix $\widetilde{{M}}$ with rational entries, such that $ | ({M} - \widetilde{{M}})_{ij} | \leq \epsilon$, in time polynomial in $\log(1/\epsilon)$. 
\item There is a Turing machine that, acting on input $x=x_1x_2\dots x_n$, outputs a classical description of $C_x$ in time bounded by a polynomial in $|x|$.
\end{itemize}

\subsection{Acceptance criterion} \label{accept}

Now that we have defined a uniform family of circuits, we need to discuss the acceptance criterion. In quantum computation it is known that performing intermediate measurements during the computation does not increase the computational power. So, without loss of generality, all measurements can be postponed until the end of the computation. A quantum computer can be defined to accept an input string $x$ if the outcome of a computational basis measurement on the first outcome qubit is $\ket{0}$. 
In a general theory, it need not be the case that all measurements can be postponed until the end of the computation without loss of generality, hence the acceptance criterion should reflect this.

The way in which a generalised probabilistic theory solves a problem might be imagined as follows. First, given the input string $x$, the circuit $C_x$ is designed and built by composing gates from the fixed finite gate set sequentially and in parallel according to the description. 
An example of such a circuit is depicted below.
$$
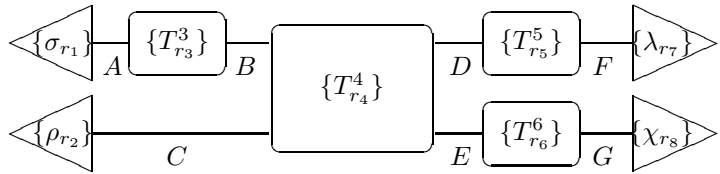
\begin{xy}
(24,7)*++[0][F-:<3pt>]{\text{$\{T^3_{r_3}\}$}}="1"; 
(47,1)*++++++[0][F-:<3pt>]{\text{$\{T^4_{r_4}\}$}}="2";
(70.6,7)*++[0][F-:<3pt>]{\text{$\{T^5_{r_5}\}$}}="3"; 
(70.6,-5)*++[0][F-:<3pt>]{\text{$\{T^6_{r_6}\}$}}="4";
"1";(12.8,7)**\dir{-};
(9,7)*{\text{$\{\sigma_{r_1}\}$}};
(12.8,-5);(36,-5)**\dir{-};
(24,-8)*{\text{$C$}};
(12.8,12);(12.8,2)**\dir{-};
(2,7);(12.8,12)**\dir{-};
(2,7);(12.8,2)**\dir{-};
(9,-5)*{\text{$\{\rho_{r_2}\}$}};
(12.8,0);(12.8,-10)**\dir{-};
(2,-5);(12.8,0)**\dir{-};
(2,-5);(12.8,-10)**\dir{-};
(15.5,4)*{\text{$A$}};
"1";(36,7)**\dir{-};
(58,7);(64.2,7)**\dir{-};
(61.3,4)*{\text{$D$}};
(80,4)*{\text{$F$}};
(80,-8)*{\text{$G$}};
(61.3,-8)*{\text{$E$}};
(33,4)*{\text{$B$}};
(87.7,7)*{\text{$\{\lambda_{r_7}\}$}};
"3";(84,7)**\dir{-};
(84,12);(84,2)**\dir{-};
(95,7);(84,12)**\dir{-};
(95,7);(84,2)**\dir{-};
(87.7,-5)*{\text{$\{\chi_{r_8}\}$}};
"4";(84,-5)**\dir{-};
(84,-10);(84,0)**\dir{-};
(95,-5);(84,-10)**\dir{-};
(95,-5);(84,0)**\dir{-};
(58,-5);(64.2,-5)**\dir{-};
\end{xy}$$
Once the circuit is built, the computation can be run. At the end of a run, each gate has a classical outcome associated with it, 
where the theory defines a joint probability for these outcomes. For the example above, the joint probability is given by
$$P(r_1, \dots ,r_8)=(\chi_{r_8}|(\lambda_{r_7}|\big(T^6_{r_6}\otimes T^5_{r_5}\big)T^4_{r_4}\big(T^3_{r_3}\otimes I_C\big)|\rho_{r_2})|\sigma_{r_1}).$$
Denoting the string of observed outcomes by $z=r_1\ldots r_8$, the final output of the computation will be given by a function $a(z)\in \{0,1\}$, where there must exist a Turing machine that computes $a$ in time polynomial in the length of the input $|x|$. The probability that a computation accepts the input string $x$ is therefore given by
$$P_x({\mathrm{accept}})=\sum_{z|a(z)=0}P(z),$$ where the sum ranges over all possible outcome strings of the circuit $C_x$.

\subsection{Efficient computation} \label{unif} 

The class of problems that can be solved efficiently in a generalised probabilistic theory can be defined as follows.
\begin{define}
For a generalised probabilistic theory $\bold{G}$, a language $\mathcal{L}$ is in the class $\bold{BGP}$ if there exists a poly-sized uniform family of circuits in $\bold{G}$, and an efficient acceptance criterion, such that
\begin{itemize}
\item $x\in\mathcal{L}$ is accepted with probability at least $\frac{2}{3}$. 
\item $x\notin\mathcal{L}$ is accepted with probability at most $\frac{1}{3}$. \end{itemize}  
\end{define} 
As ever, the choice of the constant $2/3$ is arbitrary. Any fixed constant $k$, $1/2 <k < 1$ would serve equally well.\footnote{Note that each uniform circuit (with an efficient acceptance condition) defines a random variable that maps circuit outcomes to the set $\{\textbf{accept,} \ \textbf{reject}\}$ and so one can regard multiple repetitions of a computation as a collection of i.i.d random variables (independence follows from the definition of the probabilistic structure; specifically that the sequential or parallel composition of two events corresponding to outcomes of closed circuits define independent probability distributions). This fact is independent of the form of a particular theory and so holds true for all theories in the framework. Taking this fact in conjunction with the definition of $\textbf{BGP}$ and applying the Chernoff bound provides the required result. See \cite[p.154]{Nielsen} for more discussion of the quantum case.}  

For a specified $\bold{G}$, the class $\bold{BGP}$ is the natural analogue of $\bold{BPP}$ for probabilistic classical computation, and $\bold{BQP}$ for quantum computation. Indeed, $\bold{BGP}$ reduces to $\bold{BPP}$ or $\bold{BQP}$ in the case that the theory $\bold{G}$ is in fact the classical or quantum theory. See, e.g., \cite{kit} for a proof that quantum circuits with mixed states and CP maps are equivalent in computational power to standard quantum circuits with pure states and unitary transformations.

Note that the way in which the acceptance criterion is defined implies that $\bold{P}\subseteq\bold{BGP}$, for (almost) every generalised probabilistic theory $\bold{G}$. This is a consequence of the fact that the final output is a function $a(z)$ of the string of observed events, and the only constraint on $a$ is that it can be efficiently computed by a Turing machine. Degenerate cases provide exceptions to this -- consider, e.g., any theory such that all transformations are deterministic, i.e., the outcome set of any circuit is the singleton set. One could remove these degenerate cases by generalising the acceptance function $a(.)$ so that it depend on both the outcome string $z$ and the input string $x$. Of course, the fact that $\bold{P}\subseteq\bold{BGP}$ does not have much to do with the intrinsic computational power of a GPT, but is an artefact of the acceptance criterion -- it might be interesting to weaken this criterion so that computation in theories intrinsically weaker than classical can be explored. 

\subsection{Upper bounds on computational power}\label{upperboundscomppower}  

Using the above definitions of uniform circuit families, and acceptance of an input, the following upper bound on the computational power of any generalised probabilistic theory can be obtained. The main assumption -- in addition to those involved in uniformity -- is that tomographic locality holds. The result does not require the causality assumption.
\begin{thm} \label{5}
For any generalised probabilistic theory  $\bold{G}$ satisfying tomographic locality, $\bold{BGP}\subseteq\bold{AWPP}\subseteq\bold{PP}\subseteq\bold{PSPACE}$.
\end{thm}
Here, $\bold{PSPACE}$ consists of those problems that, roughly speaking, can be solved by a classical computer using a polynomial amount of memory. $\bold{PP}$ stands for Probabilistic Polynomial time, which roughly speaking, contains those problems that can be solved by a probabilistic classical computer that must get the answer right with probability $>1/2$. The probability does not need to be bounded away from $1/2$, indeed may be greater than $1/2$ only by an exponentially small amount, hence $\bold{PP}$ contains problems that are not thought to be efficiently solvable by a classical random computer. $\bold{AWPP}$ stands for Almost-Wide Probabilistic Polynomial time, and it is known that $\bold{AWPP}\subseteq \bold{PP}$. The best known upper bound for the class of efficient quantum computations similarly states that $\bold{BQP} \subseteq \bold{AWPP}$. 

To define the class $\bold{AWPP}$, the notion of a $\bold{GapP}$ function must be introduced. Given a polynomial-time non-deterministic Turing machine $M$ and input string $x$, denote by $M_{acc}(x)$ the number of accepting computation paths of $M$ given input $x$, and by $M_{rej}(x)$ the number of rejecting computation paths of $M$ given $x$. A function $f:\{0,1\}^{*}\rightarrow\mathbb{Z}$ is a $\bold{GapP}$ function if there exists a polynomial-time non-deterministic Turing machine $M$ such that $f(x)=M_{acc}(x)-M_{rej}(x)$ for all input strings $x$. The class $\bold{AWPP}$ can be defined as follows \cite{Fenner}.
\begin{define} \label{AWPP}
The class $\bold{AWPP}$ consists of those languages $\mathcal{L}$ such that there exists a $\bold{GapP}$ function $f$, and a polynomial $r$ such that 
\begin{itemize}
\item If $x\in\mathcal{L}$ then ${2}/{3}\leq {f(x)}/{2^{r(|x|)}}\leq 1,$ 
\item if $x\notin\mathcal{L}$ then $0\leq {f(x)}/{2^{r(|x|)}}\leq{{1}/{3}}$.
\end{itemize}
\end{define} 
Once the appropriate definitions for generalised probabilistic theories are in place, the proof of Theorem~\ref{5} is a fairly straightforward extension of similar proofs for the quantum case, and is presented in Appendix~\ref{first}. 

Although formal proofs are relegated to appendices, it is useful to sketch the proof that $\bold{BGP}\subseteq\bold{PSPACE}$ in order to provide intuition about how the physical principles underlying generalised probabilistic theories lead to computational bounds.
\begin{proof}[Sketch proof.]
Consider a general circuit ${C}_T$, with $q(|T|)$ gates. Tensoring these gates with identity transformations on systems on which they do not act, and padding them with rows and columns of zeros, results in a sequence of square matrices $M^{r_q, q}, \dots , M^{r_1, 1}$, where $M^{r_n,n}$ is the matrix representing the $r_n^{\mathrm{th}}$ outcome of the $n^{\mathrm{th}}$ gate. This can be done in such a way that the probability for outcome $z=r_1\dots r_q$, is given by
$$
b^T {M}^{r_{q},q}\cdots{M}^{r_2, 2}M^{r_1,1} b
=\sum_{\{i_1,\dots,i_{q-1}\}}M^{r_{q}, q}_{1i_{q-1}}\cdots{M}_{i_{2}i_{1}}^{r_2, 2}M^{r_1, 1}_{i_11} 
$$
where $b$ is the vector $b=(1,0,\dots,0)$ and $b^T$ is its transpose. The output probability is a sum of exponentially many terms, but each term is a product of polynomially many numbers, each of which can be efficiently calculated. So a classical Turing machine can calculate each term in the sum, one after the next, keeping a running total. This requires only polynomial-sized memory.
\end{proof} 

This proof relies on the ability to decompose the acceptance probability of the computation in a form reminiscent of a (discrete) Feynman path integral. This is a consequence of the fact that transformations in a generalised probabilistic theory are linear, and thus have a matrix representation. It is pertinent then to ask where this linearity comes from. When we introduced generalised probabilistic theories in Section~\ref{intro}, we associated states (respectively, effects) with functions taking effects (respectively, states) to probabilities. As one can take linear combinations of such functions, this induces a linear structure on the set of states (respectively, effects). Thus the linear structure of generalised probabilistic theories arises from the requirement that a physical theory should be able to give probabilistic predictions about the occurrence of possible outcomes. 

Aside from linearity, a further requirement of the proof is the ability to compute efficiently the entries in the matrices representing the transformations applied in parallel in a specific circuit. Section~\ref{LT} noted that in a theory satisfying tomographic locality, a transformation $\mathcal{E}\in\bold{Transf(A,B)}$ is completely specified by its action on $\bold{St(A)}$, and so the matrix representing transformations applied in parallel can be easily calculated by taking the tensor product of the matrices representing each individual transformation. This is not the case in a theory without tomographic locality, where the tensor product structure disappears. If a transformation from $\bold{A}$ to $\bold{B}$ acts on one half of a system $\bold{AC}$, there may be no simple way to relate the linear map $\bold{\mathrm{St}(AC)}\rightarrow\bold{\mathrm{St}(BC)}$ to the action of the transformation when it is applied to a system $\bold{A}$ on its own, or indeed to a joint system $\bold{AC'}$. There may therefore be no efficient way of computing matrix elements corresponding to a transformation considered as part of a circuit of arbitrary size. An interesting direction for future work might be to weaken the assumption of tomographic locality such that the results still go through. Real Hilbert space quantum theory provides an example of a theory without tomographic locality for which the above bounds hold, since there is an efficient way of calculating relevant matrix entries.

\section{Post-selection and generalised probbilisitic theories} \label{post}   

In \cite{Post} Aaronson introduced the notion of \emph{post-selected} quantum circuits. These are quantum circuits which, in addition to having a specified qubit, on which a computational basis measurement will be made to provide the outcome, have an additional qubit on which a measurement can be performed such that we can post-select on the outcome. Instead of sampling the measurement result $r$ directly from the computational outcome qubit according to the distribution $P(r)$, only those runs of the computation are counted for which a measurement on the post-selected qubit yields the outcome $s=0$. The outcome distribution for the computation is taken to be the conditional distribution $P(r|s=0)$. 
An extra technical condition is needed, which is that there exists a constant $D$ and polynomial $w$ such that $P(S=0)\geq 1/D^{w(|x|)}$, i.e., we can only post-select on at most exponentially-unlikely outcomes.\footnote{This extra condition was missing from Aaronson's original paper on $\bold{PostBQP}$, but is needed for the definition of $\bold{PostBQP}$ to be independent of a choice of quantum gate set; see Section~2.5 of \cite{Greg}. We thank Scott Aaronson for some very interesting discussions concerning this point.}
\begin{define}
A language $\mathcal{L}$ is in the class $\bold{PostBQP}$ if there is a polynomially-sized uniform quantum circuit family, where each circuit has a computational outcome qubit and a post-selected qubit, such that when computational basis measurements are performed on these qubits, with respective outcomes $r$ and $s$,
\begin{itemize}
\item There exists a constant $D$ and polynomial $w$ such that $P(s=0)\geq 1/D^{w(|x|)}$
\item If $x\in\mathcal{L}$ then $P(r=0|s=0)\geq\frac{2}{3}$
\item If $x\notin\mathcal{L}$ then $P(r=0|s=0)\leq\frac{1}{3}$
\end{itemize}  
\end{define}
Aaronson showed in \cite{Post} that $\bold{PostBQP}=\bold{PP}$. Combining this with Theorem~\ref{5} gives
\begin{thm}
For any generalised probabilistic theory $\bold{G}$, $\bold{BGP}\subseteq\bold{PostBQP}$.
\end{thm}
Roughly speaking, a post-selecting quantum computer can simulate computation in any other theory satisfying tomographic locality. One can also define a notion of generalised circuits with post-selection on at most exponentially-unlikely outcomes. These are poly-sized uniform circuits in a generalised probabilistic theory, where the probability of acceptance is conditioned on the circuit outcome $z$ lying in a (polytime computable) subset of all possible values of $z$. Defining the class $\bold{PostBGP}$ in the obvious way, one then obtains
\begin{thm} \label{yeow}
For any generalised probabilistic theory $\bold{G}$, $\bold{PostBGP}\subseteq\bold{PP}$.
\end{thm}
The proof is in Appendix~C. Combining this with Aaronson's result yields
\begin{corr} \label{aaronson}
For any generalised probabilistic theory $\bold{G}$, $\bold{PostBGP}\subseteq\bold{PostBQP}$. 
\end{corr}

So, in a world in which we can post-select on at most exponentially-unlikely events, quantum theory is optimal for computation in the space of all tomographically local theories. 
Note that the class of problems efficiently solvable on a probabilistic classical computer with the power of post-selection is unlikely to be as large as $\bold{PP}$: it was shown in \cite{IQP} that if this class, denoted $\bold{BPP_{path}}$, is equal to $\bold{PP}$, then the polynomial hierarchy collapses to the third level. 

It was suggested in \cite{Jon} (see also \cite{markus}) that quantum theory in some sense achieves an optimal balance between the sets of available states and dynamics, in such a way that quantum theory is optimal, or at least powerful, for computation relative to the class of generalised probabilistic theories. 
It is interesting to ask whether Corollary~\ref{aaronson} can be seen as \emph{evidence} in favour of this idea. The following considerations show that caution is needed. 
Consider, for example,  the class $\bold{IQP}$ \cite{IQP}, of restricted quantum computations where the only gates allowed in a circuit are diagonal in the $\{\ket{+},\ket{-}\}$ basis. Clearly $\bold{IQP}\subseteq\bold{BQP}$, but it is unlikely that $\bold{BQP}\subseteq\bold{IQP}$. However, it was shown in \cite{IQP} that $\bold{PostIQP}=\bold{PP}=\bold{PostBQP}$. 
Alternatively, consider the class of restricted quantum computations $\bold{DQC_k}$, discussed in \cite{Clean}, known as the \emph{one clean qubit model}, where the inputs to each circuit are restricted to be one pure qubit with as many maximally mixed qubits as desired. At the end of the computation, $k$ qubits are measured in the computational basis. Clearly, $\bold{DQC_k}\subseteq\bold{BQP}$, but again, $\bold{DQC_k}$ is not believed to be universal for quantum computation.\footnote{In fact, under reasonable assumptions, $\bold{DQC_k}$ is provably not universal for quantum computation \cite{Mixed}.} It was shown in \cite{Clean} that $\bold{PostDQC_k}=\bold{PP}=\bold{PostBQP}$ for $k\geq{3}$. So, while $\bold{PostBQP}\subseteq\bold{PostDQC_k}$, under reasonable assumptions \cite{Mixed} it is not the case that $\bold{BQP}\subseteq\bold{DQC_k}$.

\section{Oracles} \label{oracle} 

In classical computation, an \emph{oracle} is a total function $O:\mathbb{N}\rightarrow\{0,1\}$. A number $x$ is said to be in an oracle $O$ if $O(x)=1$, hence oracles can decide membership in a language. Let $\bold{C}$ and $\bold{B}$ be complexity classes, then $\bold{C}^{\bold{B}}$ denotes the class $\bold{C}$ with an oracle for $\bold{B}$ (see \cite{Comp} for formal definitions). We can think of $\bold{C}^{\bold{B}}$ as the class of languages decided by a computation which is subject to the restrictions and acceptance criteria of $\bold{C}$, but allowing an extra kind of computational step: an oracle for any desired language $\mathcal{L}\in\bold{B}$ that may be queried at any stage in the course of the computation, with each such query counting as a single computational step. That is, bit strings may be generated at any stage of the computation and presented to the oracle, which in a single step, returns the information of whether the bit string is in $\mathcal{L}$ or not. Given two complexity classes, $\bold{C}_1$ and $\bold{C}_2$, we say the relation\footnote{The $=$ can be replaced with $\neq$, $\subseteq$ or $\supseteq$ equally well.} $\bold{C}_1=\bold{C}_2$ holds relative to the oracle $\bold{B}$, if $\bold{C}_1^{\bold{B}}=\bold{C}_2^{\bold{B}}$. Such a result is referred to as a \emph{relativised separation} result.

Oracles play a special role in quantum computation, forming the basis of most known computational speed ups over classical computation \cite{Nielsen}. In quantum computation, oracle queries are represented by a family $\{R_n\}$ of quantum gates, one for each query length. Each $R_n$ is a unitary transformation acting on $n+1$ qubits, whose effect on the computational basis is given by 
$$R_n\ket{x,a}= \ket{x,a\oplus{A(x)}}$$ for all $x\in\{0,1\}^m$ and $a\in\{0,1\}$, where $A$ is some Boolean function that represents the specific oracle under consideration. One could also consider more general oracles that, when queried, apply some general unitary transformation to the query state, but here, we only consider oracles that compute Boolean functions. 
In the state vector formalism of quantum theory, the action of a unitary oracle is defined on a maximal set of pure and perfectly distinguishable states, namely the computational basis. Linearly extending this to all states in the Hilbert space uniquely defines the action of the oracle on any state. 

As pointed out to us by Howard Barnum \cite{Howard},
the situation for generalised probabilistic theories is more subtle. 
Consider, for example, the density matrix formulation of quantum theory, and suppose that oracle queries correspond to a family of trace-preserving completely-positive maps $\{\mathcal{E}_n\}$. Analogously to the state vector formalism, define the action of the oracle on a maximal set of pure and perfectly distinguishable states $\{\rho_i\}_{i=1}^N$, where each $\rho_i$ is a density matrix, by
\begin{equation}\label{oracledef}
\mathcal{E}_n\big(\rho_x\otimes\rho_a\big)=\rho_x\otimes\rho_{a\oplus{A}(x)},
\end{equation}
where $\rho_x=\rho_{x_1}\otimes\dots\otimes\rho_{x_n}$ and $A$ is the function computed by the oracle. Note that
$$\rho_x\otimes\rho_a\rightarrow\rho_x\otimes\rho_{a\oplus{A}(x)} \quad \Longleftrightarrow \quad \ket{x,a}\rightarrow e^{i\phi(x,a)}\ket{x,a\oplus{A(x)}},$$
where $a=1,\dots,N$ and $e^{i\phi(x,a)}$ is some phase factor that depends on the query state. Now, in addition to being able to compute the function $A$, a quantum computer with access to the oracle may also acquire information about the function $\phi$, which may be hard to compute \cite{phaseoracle}.
The usual definition of a quantum oracle therefore prevents `sneaking in information' through phase factors. 

In generalised probabilistic theories (with sufficient distinguishable states), it is easy to produce a definition of an oracle analogous to that of Eq.~(\ref{oracledef}). But for a system type $A$, a maximal set of pure and perfectly distinguishable states does not in general span the vector space $\bold{V_A}$. Hence the action of an oracle on such a set of states will not, in general, uniquely define its action on an arbitrary state in the state space. It is then not clear what extra condition must be placed on the oracle, first to define its action on arbitrary input states, and second to prevent non-trivial information being obtained through its action on non-basis input states (perhaps via a generalised notion of phase \cite{phase}). 

Rather than attempt to solve this problem, we will instead consider a notion of `classical oracle' that can be defined in any generalised probabilistic theory that satisfies the causality assumption of Section~\ref{causality}. The causality assumption allows the construction of adaptive circuits without paradox (see \cite{Pavia1} for a more thorough discussion of the causality assumption, adaptive circuits, and conditioned transformations). In an adaptive circuit, the choice of which test to perform can depend on the outcomes $r_1,\ldots , r_k$ of previous tests in the circuit. An oracle $A:\mathbb{N}\rightarrow\{0,1\}$ defines an extra gate that can be used in a computation in addition to those of the finite gate set, but with input and output that are classical wires, rather than being typed as with the gates intrinsic to the theory.  The input to the oracle is a function $f(r_1,\ldots,r_k)$ of the outcomes of tests that appear in the circuit prior to the use of the oracle. The design of that portion of the circuit that is subsequent to the oracle can depend on the output $A(f)$ of the oracle. An oracle can be used in this way an unlimited number of times in a circuit, with each use counting as one gate.
The uniformity condition must be extended, so that for each use of the oracle in a circuit, the input $f(r_1,\ldots,r_k)$, and the design of the circuit subsequent to the oracle, are computable in poly-time by a Turing machine with access to an oracle for $\bold{A}$. The acceptance criterion can also be extended so that for a circuit outcome $z$, the function $a(z)$ is computable in poly-time by a Turing machine with access to an oracle for $\bold{A}$.
\begin{define} 
For each causal generalised probabilistic theory $\bold{G}$, a language $\mathcal{L}$ is in the class $\bold{BGP}_{cl}^{\bold{A}}$ if there exists a poly-size uniform family of circuits with access to the classical oracle $\bold{A}$, and an efficient acceptance condition, such that 
\begin{itemize} 
\item $x\in\mathcal{L}$ is accepted with probability at least $\frac{2}{3}$. 
\item $x\notin\mathcal{L}$ is accepted with probability at most $\frac{1}{3}$ 
\end{itemize}
\end{define}

We can use the notion of classical oracle to obtain the following relativised separation result.

\begin{thm} \label{NP}  
There exists a classical oracle $\bold{A}$ such that for any causal generalised probabilistic theory $\bold{G}$, $\bold{NP}^{\bold{A}}\nsubseteq\bold{BGP}_{cl}^{\bold{A}}$.
\end{thm}
The proof is in Appendix~\ref{last}. This generalises the results of \cite{AWPP} from quantum theory to causal generalised probabilistic theories that satisfy tomographic locality. 
The result proved in the appendix is actually stronger: there exists a classical oracle $\bold{A}$ such that for any causal generalised probabilistic theory $\bold{G}$ that satisfies tomographic locality, the polynomial time hierarchy is infinite and $\bold{BGP}_{cl}^{\bold{A}}\subseteq\bold{P}^{\bold{A}}$. The oracle in question is the same oracle that was used by Fortnow and Rogers in \cite{AWPP}.

\section{Discussion and conclusion}

This work has investigated the relationship between computation and physical principles. Using the circuit framework approach to generalised probabilistic theories, introduced by Hardy in \cite{Hardy1, Hardy2} and Chiribella, D'Ariano and Perinotti in \cite{Pavia1,Pavia2}, the computational power of theories formulated in operational terms can be investigated, along with the role played by simple information-theoretic or physical principles that a theory may or may not satisfy. A rigorous model of computation can be defined that allows a definition of the complexity class of problems efficiently solvable by a specific theory. The strongest known inclusion for the quantum case, $\bold{BQP}\subseteq\bold{AWPP}$, which implies $\bold{BQP} \subseteq \bold{PP}\subseteq\bold{PSPACE}$, still holds in any theory satisfying tomographic locality, and it is notable that this includes even those theories that violate the causality principle. Combining these results with a result of Aaronson's, it follows that any problem efficiently solvable in a theory satisfying tomographic locality can also be solved efficiently by a post-selecting quantum computer. In fact, one can say something stronger: any problem efficiently solvable with post-selection in a theory satisfying tomographic locality can also be solved efficiently by a post-selecting quantum computer. Roughly speaking, then, in a world with post-selection, quantum theory is optimal for computation in the space of all tomographically local theories.




We discussed the problem of defining a computational oracle for an arbitrary theory. In general, this problem may have no good solution, if it is required that the definition of an oracle reduce to the standard definition in the quantum case. Nonetheless, a notion of `classical oracle' can be defined in any theory that satisfies the causality principle, and for such theories there exists a classical oracle relative to which $\bold{NP}$ is not contained in $\bold{BGP}$. 
It is plausible that there is an interesting subclass of theories, for which a notion of oracle can be defined that admits `superposition' of inputs, and reduces to the standard definition in the quantum case. If so, then for these theories, the solution of the `subroutine problem' of \cite{Sub} might serve as an interesting computational principle that could rule out certain theories, potentially providing a new principle from which quantum theory can be derived.  

 
An open question is to establish tighter bounds on the power of general theories. Even with tomographic locality assumed, there is a lot of freedom in the construction of a generalised theory. Is there an explicit construction that solves a hard problem, that is, a problem at least thought to be hard for quantum computers? Even better, can we describe a complexity class, potentially larger than $\bold{BQP}$, and an explicit construction of a general theory $\bold{G}$, such that this class is contained in $\bold{BGP}$? It would be interesting to determine whether violation of the causality principle can confer extra computational power. An initial thought is that there could be a non-causal theory that can efficiently solve \textbf{NP}-complete problems. Given that the inclusion $\bold{BGP}\subseteq\bold{AWPP}$ holds even for non-causal (tomographically local) theories, however, this can only be the case if $\bold{NP}$ is contained in $\bold{AWPP}$. At present, this is unknown, and establishing the question either way would constitute a major advance in complexity theory. Still, it would be interesting if the violation of causality enabled the efficient solution of other problems, thought to be hard for quantum computers, but known to be in $\bold{AWPP}$.

Finally, although our main results do not require the causality principle, we have nonetheless been considering circuits in which gates appear in a fixed structure. It would be interesting to investigate the computational power of theories in which there is no such definite structure. Frameworks for describing situations with indefinite causal structure have been defined with the aim of discussing aspects of quantum gravity \cite{Grav, lucien}. Some preliminary remarks on the computational power of such theories were given in \cite{lucien, PaviaNoCause} and a specific query complexity problem that can be solved with fewer queries on a quantum computer in which the gates do not appear in a fixed order than on a standard quantum computer was presented in \cite{BruknerNoCause}.

 
\subsection*{Acknowledgements}

The authors thank Howard Barnum for helpful discussions on the definition of oracles in GPTs. CML thanks Lance Fortnow for a useful email correspondence about results presented in \cite{AWPP} and Scott Aaronson for a useful email correspondence about \cite{Post}. CML also thanks John Selby, John-Mark Allen, Matty Hoban and Ray Lal for helpful discussions and John Selby and John-Mark Allen for proof reading a draft of the current paper.  This work was supported by the FQXi Large Grants \emph{Time and the structure of quantum theory}, and \emph{Thermodynamic vs information theoretic entropies in probabilistic theories}. This work was supported by the CHIST-ERA DIQIP project.

\begin{appendices}

\section{Approximate circuit families}  

Consider a poly-size uniform circuit family $\{C_x\}$, defined over a finite gate set $\mathcal{G}$. Each gate in $\mathcal{G}$ corresponds to some finite set of transformations, one for each classical outcome of the gate. From the uniformity condition, the entries of the matrices representing these transformations can be calculated to accuracy $\epsilon$ in time poly(log($1/\epsilon$)). With $\epsilon(|x|)$ a function of the input size, consider a family $\{\widetilde{C_x}\}$ of approximations to the original circuits, where matrix elements are replaced by rational numbers within $\epsilon(|x|)$ of the original matrix elements. Call $\{\widetilde{C_x}\}$ an \emph{$\epsilon(|x|)$-approximation} to $\{C_x\}$. The following result shows that $\{\widetilde{C_x}\}$ can simulate $\{C_x\}$, to an accuracy dependent on $\epsilon(|x|)$. 
\begin{prop}\label{3} 
Let $\{C_x\}$ be a uniform circuit family, with the number of gates in $C_x$ bounded by a polynomial $q(|x|)$. Let $\{\widetilde{{C_x}}\}$ be an $\epsilon(|x|)$-approximation to $\{C_x\}$, with $\epsilon(|x|)\leq 1$. If the circuit $C_T\in\{C_x\}$ gives an outcome sequence $z$ with probability $P(z)$, then the circuit $\widetilde{C_T}\in\{\widetilde{C_x}\}$ gives outcome sequence $z$ with amplitude $\widetilde{P}(z)$ such that 
$$ 
|P(z)-\widetilde{P}(z)| \leq  D^{q(|T|)-1} q(|T|)\epsilon(|T|)N , 
$$
where $N$ and $D$ are constants depending on the gate set $\mathcal{G}$.
\end{prop} 
The word \emph{amplitude} here should not be confused with the complex amplitudes of quantum theory. It is used for the real-valued quantity which approximates an outcome probability for the original circuit family, and is used rather than the term \emph{probability}, because this quantity can be (slightly) less than $0$ or (slightly) greater than $1$. (The approximating circuit family is a mathematical construction that need not correspond precisely to a valid circuit family in the theory.) This proposition will be useful in the main proofs, since if $\{C_x\}$ is a circuit family that decides some language $\mathcal{L}$ in $\bold{BGP}$, it follows that a $\frac{1}{12q(|x|)D^{q(|x|)-1}N}$-approximation to $\{C_x\}$ will accept a string $x\in \mathcal{L}$ with amplitude at least $7/12$, and will accept a string  $x\notin \mathcal{L}$ with amplitude at most $5/12$, hence the success amplitude is still bounded away from $1/2$. The uniformity condition ensures that such an $\epsilon(|x|)$-approximation can be constructed in time polynomial in $|x|$.

In order to prove the proposition, two lemmas will be helpful.


\begin{lem}\label{1}
Let $M$ be a real $n\times{m}$ matrix such that for each entry, $m_{ij}$, we have that $|m_{ij}|\leq\epsilon$, for $\epsilon>0$. Then $$\Vert M\Vert _{op} \leq{nm}\epsilon,$$ where $\Vert.\Vert _{op}$ is the operator norm.
\end{lem}
\begin{proof}
Let $M_i$ be the $i^{th}$ row of $M$. Then
$$ 
\begin{aligned} 
| M_i |_E =\sqrt{\sum_{j=1}^m{m^2_{ij}}} \leq\sum_{j=1}^m{|m_{ij}|}\leq{\epsilon}m, 
\end{aligned} 
$$
where $|.|_E$ is the Euclidean norm, hence
$$ \begin{aligned} | M  v |_E \leq\sum_{i^=1}^n |  M_i v |
\leq{\sum_{i=1}^n\epsilon{m}}={nm}\epsilon,
\end{aligned} 
$$                                                                             
for $| v | =1$, where the second inequality follows from the Cauchy-Schwarz inequality. Thus  $\Vert M \Vert _{op} \leq  {nm}\epsilon.$
\end{proof}

\begin{lem} \label{2}
Let $\{{M}_i\}_{i=1}^T$ and $\{\widetilde{{M}}_i\}_{i=1}^T$ be two sets of matrices. Then the $T$-fold product of these matrices satisfies
$$ \begin{aligned} \Vert{M}_T&\dots{M}_1-\widetilde{{M}}_T\dots\widetilde{{M}}_1\Vert_{op}\leq D^{T-1}\sum_{i=1}^T\Vert{M}_i-\widetilde{{M}}_i\Vert_{op}, \end{aligned}$$
where $D=\max\{\Vert M_1\Vert_{op}, \dots, \Vert M_T\Vert_{op}, \Vert \widetilde{M}_1\Vert_{op}, \dots, \Vert\widetilde{M}_T\Vert_{op}\}$.
\end{lem}
 \begin{proof}
Consider the case of $T=2$. With $|v|=1$, 
$$\begin{aligned} &\vert\big({M}_2{M}_1-\widetilde{{M}}_2\widetilde{{M}}_1\big)v\vert_E \\ &=\vert\big({M}_2{M}_1-\widetilde{{M}}_2{{M}}_1\big)v+\big(\widetilde{{M}}_2{M}_1-\widetilde{{M}}_2\widetilde{{M}}_1\big)v\vert_E \\
&\leq\vert\big({M}_2-\widetilde{{M}}_2\big){{M}_1}v\vert_E+\vert\widetilde{{M}}_2\big({M}_1-\widetilde{{M}}_1\big)v\vert_E \\
&\leq\Vert{M}_2-\widetilde{{M}}_2\Vert_{op}\Vert M_1\Vert_{op}+\Vert\widetilde{M}_2\Vert_{op}\Vert{M}_1-\widetilde{{M}}_1\Vert_{op}.
\end{aligned}$$
Thus
$$\Vert{M}_2{M}_1-\widetilde{{M}}_2\widetilde{{M}}_1\Vert_{op}\leq D\Vert{M}_1-\widetilde{{M}}_1\Vert_{op}+D\Vert{M}_2-\widetilde{{M}}_2\Vert_{op}$$
The result follows from induction on $T$.
\end{proof} 

We can now prove Proposition~\ref{3}.
\begin{proof} 
A particular outcome sequence of the circuit ${C_T}\in\{{C_x}\}$ corresponds to a sequence of matrices  ${\mathcal{G}}^{r_{1},1},  \dots, {\mathcal{G}}^{r_{q},q}$, where ${\mathcal{G}}^{r_i,i}$ represents the $r_i$th outcome of the $i$th gate in ${C_T}$. Note that states and effects are included in this sequence. Tensoring these gates with identity transformations on systems on which they do not act and padding the corresponding matrices with rows and columns of zeros results in a sequence of square matrices ${{M}}^{r_{q},q}, \dots, {{M}}^{r_{1},1}$ such that
\[
P(z) = P(r_1,\dots , r_q) = b^T.{{M}}^{r_{q},q}\dots{{M}}^{r_1,1}.b,
\]
where $b$ is the vector $(1,0,\dots,0)$ and $b^T$ is its transpose. Similarly for $\widetilde{\mathcal{G}}^{r_{1},1},  \dots, \widetilde{\mathcal{G}}^{r_{q},q}$, so that
\[
\widetilde{P}(z) = \widetilde{P}(r_1,\ldots , r_q) = b^T.\widetilde{{M}}^{r_{q},q}\dots\widetilde{{M}}^{r_1,1}.b.
\]
Note that $\Vert M^{r_i,i}\Vert_{op}\leq \Vert{\mathcal{G}}^{r_{i},i}\Vert_{op}$ and  $\Vert \widetilde{M}^{r_i,i}\Vert_{op}\leq \Vert\widetilde{\mathcal{G}}^{r_{i},i}\Vert_{op}$, for all $i$.
Therefore,
$$\begin{aligned}
|P(z)-\widetilde{P}(z)|&=|b^T\big({M}^{r_{q},q}\dots{M}^{r_1,1}-\widetilde{{M}}^{r_{q},q}\dots\widetilde{{M}}^{r_1,1}\big)b|\\
&\leq\vert b^T\vert_E \vert\big({M}^{r_{q},q}\dots{M}^{r_1,1}-\widetilde{{M}}^{r_{q},q}\dots\widetilde{{M}}^{r_1,1}\big)b\vert_E \\
&\leq D'^{q(|T|)-1}\sum_{n=1}^{q}\Vert{M}^{r_n,n}-\widetilde{{M}}^{r_n,n}\Vert_{op}\leq{D'^{q(|T|)-1}q(|T|)N\epsilon(|T|)},
\end{aligned}$$
where if $n_im_i$ is the size of the matrix ${\mathcal{G}}^{r_i,i}$, then  
\[
N=\max\{n_{q}m_{q},\dots,n_1m_1\},
\]
and 
\[
D'=\max\{\Vert {\mathcal{G}}^{r_1,1}\Vert_{op}, \dots, \Vert {\mathcal{G}}^{r_q,q}\Vert_{op}, \Vert \widetilde{\mathcal{G}}^{r_1,1}\Vert_{op}, \dots, \Vert \widetilde{\mathcal{G}}^{r_q,q}\Vert_{op}\}.
\]
Note that, as circuits are built from finite gate sets, $N$ is a constant. The first inequality follows from the Cauchy-Schwarz inequality, the second from that fact that $|b^T|=1$ and lemma \ref{2}, the third from lemma \ref{1}, the fact that the sum has $q(|T|)$ entries and the fact that, as $\widetilde{C}_T$ is an $\epsilon$-approximation of $C_T$, the matrix ${M}^{r_i,i}-\widetilde{{M}}^{r_i,i}$ has entries satisfying $|m_{ij}-\widetilde{m}_{ij}|\leq\epsilon$. 

The reverse triangle inequality gives
\[
\Vert\widetilde{\mathcal{G}}^{r_i,i}\Vert_{op}-\Vert{\mathcal{G}}^{r_i,i}\Vert_{op}\leq \Vert\widetilde{\mathcal{G}}^{r_i,i}-{\mathcal{G}}^{r_i,i}\Vert_{op}\leq N\epsilon(|T|).
\]
With $\epsilon(|T|)\leq 1$, and 
\[
D''=\max\{\Vert {\mathcal{G}}^{r_1,1}\Vert_{op}, \dots, \Vert {\mathcal{G}}^{r_q,q}\Vert_{op}\},
\]
we have $D'\leq D\equiv D'' +N$, which completes the proof.

\end{proof}

\section{Proof of Theorem~\ref{5}}  \label{first}

One method of proving Theorem~\ref{5} is to use $\bold{GapP}$ functions. $\bold{GapP}$ functions were first studied in the context of quantum computation by Fortnow and Rogers in \cite{AWPP}, where, among other things, they showed that $\bold{BQP}\subseteq\bold{AWPP}$. A good discussion on $\bold{GapP}$ functions can be found in Watrous's survey of quantum complexity theory \cite{Comp1}. Proofs in this section are modifications and generalisations of proofs presented in \cite{AWPP, Comp1, Comp}.

Given a polynomial-time non-deterministic Turing machine $M$ and input string $x$, denote by $M_{acc}(x)$ the number of accepting computation paths of $M$ given input $x$, and by $M_{rej}(x)$ the number of rejecting computation paths of $M$ given $x$. A function $f:\{0,1\}^{*}\rightarrow\mathbb{Z}$ is a $\bold{GapP}$ function if there exists a polynomial-time non-deterministic Turing machine $M$ such that $f(x)=M_{acc}(x)-M_{rej}(x)$ for all input strings $x$. 

Many complexity classes can be described in terms of $\bold{GapP}$ functions. For example the class $\bold{PP}$ can be defined as those languages $\mathcal{L}$ such that, for some $\bold{GapP}$ function $f$ and any input string $x$, if $x\in\mathcal{L}$ then $f(x)>0$ but if $x\notin\mathcal{L}$ then $f(x)\leq{0}$.
A useful class of $\bold{GapP}$ functions is provided by the following theorem.
\begin{thm} \label{gap}
Any function $f:\{0,1\}^*\rightarrow\mathbb{Z}$ that can be computed in poly-time by a Turing machine is a $\bold{GapP}$ function.
\end{thm}
For a proof, see \cite[p.237]{Comp}.

The notation $\langle{x,y}\rangle$ denotes the pairing function \cite{AWPP}: that is, a poly-time computable function that maps the pair of strings $x$ and $y$ bijectively to the
set of finite length strings $\{0,1\}^*$ such that, given $\langle{x,y}\rangle$, both $x$ and $y$ can be extracted in poly-time.
The following proposition gives slight generalisations of standard closure properties of $\bold{GapP}$ functions.
%
%
\begin{prop} \label{bant}
For a polynomial $q$ and $\bold{GapP}$ function $f$, let $h:\{0,1\}^*\rightarrow\mathbb{Z}$ be defined for all $x\in\{0,1\}^*$ by
$$h(x)=\sum_{{\substack{
  |y| \leq q(|x|) \\
   y\in{L_x}
  }}}f(\langle{x,y}\rangle),$$
where $L_x$ is some set (that may depend on $x$) with the property that membership of $y$ in $L_x$ can be determined in time polynomial in $|x|$. Then $h$ is a $\bold{GapP}$ function. 

Now let $g:\{0,1\}^*\rightarrow\mathbb{Z}$ be defined for all $x\in\{0,1\}^*$ by
$$g(x)=\prod_{{\substack{
  1\leq i \leq q(|x|) \\
   i\in{L_x}
  }}} f(\langle{x,i}\rangle),$$
where the symbol $i$ appearing as the second argument on the pairing is a binary encoding of $i$ and $L_x$ is some set with the property that membership of $i$ in $L_x$ can be determined in time polynomial in $|x|$. Then $g$ is also a $\bold{GapP}$ function.
\end{prop}
\begin{proof}
We will prove the first statement only as the second statement follows from a similar generalisation of a standard argument. Let $f(x)=M_{acc}(x)-M_{rej}(x)$ for some non-deterministic poly-time Turing machine, $M$. Let $N$ be a non-deterministic poly-time Turing machine that, on input $x\in\{0,1\}^*$, guesses a string $y$ of length $\leq q(|x|)$, decides whether $y$ is in $L_x$, and 
\begin{itemize}
\item if $y\in L_x$, simulates $M$ on input $\langle{x,y}\rangle$. 
\item if $y\notin L_x$, guesses a bit $b$ and accepts if and only if $b=0$. 
\end{itemize}
$N$ runs in poly-time, and for every $x\in\{0,1\}^*$, $N_{acc}(x)-N_{rej}(x)=h(x)$, hence $h$ is a $\bold{GapP}$ function.
\end{proof}

For the rest of this section, assume that the pairing function is used whenever a function has two or more arguments. $\bold{GapP}$ functions are intimately related to computation in generalised probabilistic theories, as the following result shows. 

\begin{thm} \label{4} 
Let $\{C_x\}$ be a poly-size uniform family of circuits in a generalised probabilistic theory. Then for any polynomial $w$ and constant $D$, there exists a function $\epsilon(|x|)\leq {1}/{D^{w(|x|)}}$, and an $\epsilon(|x|)$-approximation $\{\widetilde{{C_x}}\}$ to $\{C_x\}$, such that the amplitude for acceptance\footnote{Note that, as $\{\widetilde{{C_x}}\}$ is a mathematical construction, it need not correspond to a valid circuit family in the theory and so cannot be said to accept or reject an input string. However, for ease of notation, we will say an approximating circuit `accepts' an input string if $a(z)=0$ where $z$ is the outcome sequence of that approximating circuit, and `rejects' the input string otherwise.} of a circuit $\widetilde{C_T}\in\{\widetilde{{C_x}}\}$ is given by
$$\widetilde{P}_T(\mathrm{accept})=\frac{f(T)}{2^{p(|T|)}},$$ 
where $f$ is a $\bold{GapP}$ function and $p(|T|)$ is a polynomial in the size of the input string.
\end{thm} 
\begin{proof}   
It follows from the uniformity condition that for any polynomial $w$, there is an $\epsilon(|x|)$-approximation $\{\widetilde{{C_x}}\}$ to $\{C_x\}$, with $\epsilon(|x|)\leq {1}/{D^{w(|x|)}}$, such that the entries in the matrices representing gates in the circuit $\widetilde{C_T}\in\{\widetilde{{C_x}}\}$ have rational entries, and can be computed in time polynomial in $|T|$. Furthermore, the rational entries can be taken to have the form $c/2^d$, with $c\in\mathbb{Z}$, $d\in\mathbb{N}$, and $d$ a polynomial function of $|T|$. Padding circuits with identity gates if necessary, assume that the number of gates in the circuit $\widetilde{C_T}$ is given by a polynomial function $q(|T|)$. A particular outcome of the circuit corresponds to matrices  $\widetilde{\mathcal{G}}^{r_{1},1},  \dots, \widetilde{\mathcal{G}}^{r_{q},q}$, where $\widetilde{\mathcal{G}}^{r_i,i}$ represents the $r_i$th outcome of the $i$th gate in $\widetilde{C_T}$. States and effects are included in this sequence.


By tensoring these gates with identity transformations on systems on which they do not act and padding the corresponding matrices with rows and columns of zeros, we can obtain a sequence of square matrices $\widetilde{\mathcal{M}}^{r_{1},1},\ldots, \widetilde{\mathcal{M}}^{r_{q},q}$, such that (i) rows and columns of these matrices are indexed by bit strings of length $y(|T|)$, with $y(|T|)$ a polynomial function, and (ii) the amplitude of outcome $z = r_1,\ldots, r_q$ is given by 
\[
b^T . \, \widetilde{\mathcal{M}}^{r_{q},q}\cdots \widetilde{\mathcal{M}}^{r_{1},1} \, . b,  
\]
where $b$ is the vector $(1,0,\ldots , 0)$ and $b^T$ is its transpose. 
Note that  for each $\mathcal{\widetilde{M}}^{r_i,i}$, the matrix $2^{d}\mathcal{\widetilde{M}}^{r_i,i}$ has integer entries. 

Consider the function $h:\{0,1\}^*\rightarrow\mathbb{Z}$ given by
$$h(T,r_1,\ldots,r_q,n,i_0,\dots,i_{q})=M^{r_n,n}_{i_ni_{n-1}},$$
where $i_0,\ldots,i_q$ are bit strings of length $y(|T|)$, and $M^{r_n,n}_{i_ni_{n-1}}$ is the ${i_ni_{n-1}}$ entry of the matrix $2^{d}\mathcal{\widetilde{M}}^{r_n,n}$ . 
By the uniformity condition, these matrix entries can be calculated in polynomial time by a Turing machine, so by Theorem~\ref{gap},  $h$ is a $\bold{GapP}$ function. 

The amplitude for outcome $z=r_1\ldots r_q$ is given by
\begin{eqnarray*}
\widetilde{P}(z)&=&\frac{1}{2^{d q}}\sum_{\{i_1,\dots,i_{q-1}\}}M^{r_{q},q}_{1i_{q-1}}\dots{M}_{i_{2}i_{1}}^{r_2,2}M^{r_1,1}_{i_11} \\
&=&\frac{1}{2^{d q}}\sum_{\{i_1,\dots,i_{q-1}\}}\prod_{1\leq n \leq q}h(T,r_1,\ldots,r_q,n,i_0=1,i_1,\dots,i_{q-1},i_q=1) \\
&=&\frac{1}{2^{d q}}\sum_{\{i_1,\dots,i_{q-1}\}}g(T,r_1,\ldots,r_q,i_1,\dots,i_{q-1}),\\
&=&\frac{f'(T,z)}{2^{d q}},
\end{eqnarray*}
where $g$ is a $\bold{GapP}$ function by Proposition~\ref{bant}, hence $f'$ is a $\bold{GapP}$ function by another application of Proposition~\ref{bant}. 

The amplitude for the circuit $\widetilde{C_T}$ to accept is given by 
$$\widetilde{P}_T(\mathrm{accept})=\sum_{a(z)=0}\widetilde{P}_T(z)=\sum_{a(z)=0}\frac{f'(T,z)}{2^{d q}},$$
where $a(z)$ is the function that determines if $z$ is an accepting or rejecting outcome. By the uniformity condition, $a(z)$ can be calculated in polynomial time by a Turing machine, hence Proposition~\ref{bant} gives
$$\widetilde{P}_T(\mathrm{accept})=\frac{f(T)}{2^{p(|T|)}},$$
where $f$ is a $\bold{GapP}$ function and $d(|T|)q(|T|)=p(|T|)$ is a polynomial that takes values in $\mathbb{N}$. 
\end{proof} 


The class $\bold{AWPP}$ (Almost Wide Probabilistic Polynomial time) can be defined \cite{Fenner} as follows.
\begin{define} \label{AWPP}
The class $\bold{AWPP}$ consists of those languages $\mathcal{L}$ such that there exists a $\bold{GapP}$ function $f$, and a polynomial $r$ such that 
\begin{itemize}
\item If $x\in\mathcal{L}$ then ${2}/{3}\leq {f(x)}/{2^{r(|x|)}}\leq 1,$
\item if $x\notin\mathcal{L}$ then $0\leq {f(x)}/{2^{r(|x|)}}\leq{{1}/{3}}$.
\end{itemize}
\end{define} 
The $1/3-2/3$ separation can be replaced by any constant, positive, separation \cite{Fenner}. 
\begin{thm} \label{awpp}
For any generalised probabilistic theory $\bold{G}$, $\bold{BGP}\subseteq\bold{AWPP}$.
\end{thm}
\begin{proof} 
If a language $\mathcal{L} \in \bold{BGP}$, then there is a poly-size uniform circuit family $\{C_x\}$ such that $P_x(\mathrm{accept}) \geq 2/3$ if $x\in \mathcal{L}$, and $P_x(\mathrm{accept}) \leq 1/3$ if $x\notin \mathcal{L}$. Assume that for all $x$, $1/10 \leq P_x(\mathrm{accept}) \leq 9/10$.\footnote{This can be ensured, if necessary, by considering the circuit $C_T$ to be carried out in parallel with a biased coin toss. With probability $1/5$, the coin is tails, in which case the output of the circuit is ignored, and acceptance/rejection are returned with probability $1/2$ each. Taken together, these circuits and coin tosses define a modified circuit family $\{C_x'\}$, and in the following, approximating circuit families can be assumed to be defined relative to $\{C_x'\}$. } By Theorem~\ref{4}, there is an $\epsilon(|x|)$-approximation to $\{C_x\}$ such that the amplitudes determined by the approximating family satisfy
$$\widetilde{P}_x(\mathrm{accept})=\frac{f(x)}{2^{p(|x|)}},$$ 
with $f$ a $\bold{GapP}$ function. Furthermore, for any polynomial $w$, $\epsilon(|x|)$ can be chosen so that $\epsilon(|x|)\leq {1}/{D^{w(|x|)}}$. Hence by Proposition~\ref{3}, $\epsilon(|x|)$ can be chosen small enough that $\widetilde{P}_x(\mathrm{accept})\geq 7/12$ if $x\in \mathcal{L}$ and $\widetilde{P}_x(\mathrm{accept})\leq 5/12$ if $x\notin \mathcal{L}$, and for all $x$, $0 \leq \widetilde{P}_x(\mathrm{accept}) \leq 1$. Taking $p(|x|)$ to be the function $r(|x|)$ in definition \ref{AWPP} and noting that $5/12-7/12$ is a constant, positive, separation, gives the result.
\end{proof}
     
It is well known that $\bold{AWPP}\subseteq\bold{PP}\subseteq\bold{PSPACE}$ (see, for example, \cite{FFKL} and references therein). 

\section{Proof of Theorem~\ref{yeow}} 

An alternate definition of the class $\bold{PP}$ can be stated \cite{Li, Fenner} as follows.  
\begin{define} \label{PP}
The class $\bold{PP}$ consists of those languages $\mathcal{L}$ such that there exist $\bold{GapP}$ functions $f$ and $h$ so that for all $x$
\begin{itemize}
\item If $x\in\mathcal{L}$ then ${2}/{3}\leq {f(x)}/{h(x)}\leq 1,$ 
\item if $x\notin\mathcal{L}$ then $0\leq {f(x)}/{h(x)}\leq{{1}/{3}}.$
\end{itemize}
\end{define}
The $1/3-2/3$ separation can be replaced by any constant, positive, separation \cite{Fenner}. 

In order to prove Theorem~\ref{yeow}, consider a uniform family of circuits $\{C_x\}$ in the generalised probabilistic theory $\bold{G}$. Let $S_T$ be a subset of the possible outcomes of the circuit $C_T$, with respect to which post-selection is defined, so that $P_T(\mathrm{accept}|S_T) \geq 2/3$ for $T \in \mathcal{L}$ and $\leq 1/3$ for $T \notin \mathcal{L}$.  As in the proof of Theorem~\ref{5}, assume that these probabilities are also bounded away from $0$ and $1$ so that for all $T$, $1/10 \leq P_T(\mathrm{accept}|S_T)  \leq 9/10$.\footnote{This can be done, as before, by the introduction of a biased coin parallel to the circuit. If the circuit outcome is in $S_T$ and the coin is heads, then accept or reject, depending on the circuit outcome. If the outcome is in $S_T$ and the coin is tails then accept or reject with probability $1/2$ each.}

By Theorem~\ref{4}, there is an $\epsilon(|x|)$-approximation to $\{C_x\}$ such that, in the approximating family, the joint amplitude to accept the computation and have an outcome from the set $S_T$ is
\[
\widetilde{P}_T(\mathrm{accept}, S_T)=\frac{f(T)}{2^{p(|T|)}},
\]
with $f$ a $\bold{GapP}$ function. Similarly, 
\[
\widetilde{P}_T(S_T)=\frac{g(T)}{2^{q(|T|)}},
\]
with $g$ a $\bold{GapP}$ function and $q$ a polynomial. Furthermore, for any polynomial $w$ and constant $D$, $\epsilon(|x|)$ can be chosen so that $\epsilon(|x|)\leq 1/D^{w(|x|)}$. Hence by Proposition~\ref{3} and the fact that we are post-selecting on at most exponentially-unlikely outcomes, $\epsilon(|x|)$ can be chosen small enough that for the approximating circuit family, $\widetilde{P}_T(S_T) > 0$. This means that for the approximating circuit family, the conditional
\[
\widetilde{P}_T(\mathrm{accept}|S_T) = \frac{\widetilde{P}_T(\mathrm{accept}, S_T)}{\widetilde{P}_T(S_T)},
\]
is well defined. Furthermore, $\epsilon(|x|)$ can be chosen small enough that $\widetilde{P}_T(\mathrm{accept}|S_T)\geq 7/12$ if $x\in \mathcal{L}$, $\widetilde{P}_T(\mathrm{accept}|S_T)\leq 5/12$ if $x\notin \mathcal{L}$, and using the assumption that the original circuit family probabilities are bounded away from $0$ and $1$,  the approximating amplitudes satisfy $0 \leq \widetilde{P}_T(\mathrm{accept}|S_T) \leq 1$.

Now,
\[
\widetilde{P}_T(\mathrm{accept}|S_T) = \frac{2^{q(|T|)}f(T)}{2^{p(|T|)}g(T)}  = \frac{l(T)}{h(T)},
\]
where $h(T)=2^{p(|T|)}g(T)$ and $l(x)=2^{q(|T|)}f(T)$ are $\bold{GapP}$ functions. This follows from Theorem~\ref{gap}, Proposition~\ref{bant}, and the fact that both $p$ and $q$ are polynomials taking values in $\mathbb{N}$. The result follows.
 
\section{Proof of Theorem~\ref{NP}}   \label{last}
 
Denote by $\bold{PH}$ the polynomial time hierarchy: the union of an infinite hierarchy of classes $\bold\Sigma_k$, $\bold\Delta_k$ and $\bold\Pi_k$ for $k\in\mathbb{N}$, where $\bold\Sigma_0=\bold\Delta_0=\bold\Pi_0=\bold{P}$ and $\bold\Sigma_{k+1}= \bold{NP}^{\bold\Sigma_k}$, $\bold\Delta_{k+1}= \bold{P}^{\bold\Sigma_k}$ and $\bold\Pi_{k+1}= \bold{coNP}^{\bold\Sigma_k}$. The polynomial time hierarchy is a natural way of classifying the complexity of problems beyond the class $\bold{NP}$. It is a strongly held belief in computer science that $\bold{NP}$ includes non-polynomial-time problems.

Theorem~\ref{NP} is a corollary of two results, the first of which is due to \cite{FFKL} and \cite{Y}:
\begin{thm}
There exists an oracle $\bold{A}$ such that $\bold{P}^{\bold{A}}=\bold{AWPP}^{\bold{A}}$ and the polynomial time hierarchy is infinite.
\end{thm} 
The second is that Theorem~\ref{awpp} \emph{relativizes}.
\begin{thm}
For any classical oracle $\bold{A}$ we have that $\bold{BGP}_{cl}^{\bold{A}}\subseteq\bold{AWPP}^{\bold{A}}$ for any causal $\bold{G}$.
\end{thm}  
\begin{proof} 
Given the uniformity condition for circuit families with an oracle, entries in the matrices representing gates in a circuit are all computable in polynomial time by a Turing machine with access to the oracle $\bold{A}$. Thus the proof of Theorem~\ref{4} goes through essentially unchanged, except that in this case the conclusion is that the acceptance amplitude is 
\[
\widetilde{P}_x(\mathrm{accept})=\frac{f(x)}{2^{p(|x|)}},
\]
where $p(|x|)$ is a polynomial function of the size of the input and $f$ is a $\bold{GapP^{\bold{A}}}$ function. A $\bold{GapP^{\bold{A}}}$ function is defined in a similar fashion to a $\bold{GapP}$ function, except instead of counting the difference between the number of accepting and rejecting paths for any input into a non-deterministic Turing machine, $\bold{GapP^{\bold{A}}}$ functions count the difference between the number of accepting and rejecting paths for any input into a non-deterministic Turing machine with access to the oracle $\bold{A}$. 

$\bold{AWPP}^{\bold{A}}$ can be defined with respect to $\bold{GapP^{\bold{A}}}$ functions by just replacing every mention of $\bold{GapP}$ functions with $\bold{GapP^{\bold{A}}}$ functions in Definition~\ref{AWPP}. Thus the proof that $\bold{BGP}_{cl}^{\bold{A}}\subseteq\bold{AWPP}^{\bold{A}}$, for any causal GPT and oracle $\bold{A}$, goes through exactly the same as the proof of Theorem~\ref{awpp}.
\end{proof} 
Hence we obtain
\begin{thm} 
There exists a classical oracle $\bold{A}$ relative to which $\bold{BGP}_{cl}^{\bold{A}}\subseteq\bold{P}^{\bold{A}}$, for all causal $\bold{G}$, and the polynomial time hierarchy is infinite. 
\end{thm}
This implies that there exists a classical oracle relative to which $\bold{NP}$ is not contained in $\bold{BGP}$, for any causal theory $\bold{G}$ satisfying tomographic locality. This generalises the results of \cite{AWPP} from quantum theory to general theories. 

\end{appendices}

\end{document}